%% file: main.tex
\documentclass[sigplan,10pt,nonacm]{acmart}

\setcopyright{none}
\AtBeginDocument{}

\usepackage{algorithm}
\usepackage{algpseudocode}
\usepackage{enumitem}
\usepackage{colortbl}

\hypersetup{
  pdftitle={JointServe: Serving Multi-Model Agents on Memory-Constrained GPUs},
  pdfauthor={Daijin Hu},
  pdfkeywords={Large Language Model Serving, Multi-Model Agents, GPU Memory Management, KV Cache}
}

\begin{document}

\title{JointServe: Serving Multi-Model Agents on Memory-Constrained GPUs}

\author{Daijin Hu}
\affiliation{\institution{Beihang University}
  \city{}
  \country{}}
\email{daijinhu@buaa.edu.cn}

\renewcommand{\shortauthors}{Daijin Hu}

\input{abstract}


\maketitle

\thispagestyle{plain}
\pagestyle{plain}

\input{intro}
\input{background}
\input{overview}
\input{design}
\input{implementation}
\input{evaluation}
\input{discussion}
\input{related}
\input{limitations}
\input{conclusion}

\clearpage

\input{main.bbl}
\clearpage
\appendix
\input{appendix}

\end{document}

%% file: abstract.tex
\begin{abstract}
LLM agents are evolving from a single model calling tools to planners that orchestrate multiple expert models whose capabilities live in distinct weights. This shift breaks a hidden assumption of today's serving systems—that the working set of models fits in GPU memory. On the memory-constrained GPUs typical of private and edge deployments, naive whole-model swapping serializes weight transfer with computation and, without host-backed KV restoration, discards session state at a switch.

We present JointServe, a serving system that treats weight shards and KV blocks as residency objects competing for one unified memory budget. JointServe exploits three structural properties of agent workloads—directed calls, gapped execution, and dual-state residency—to prepare models before they are called: an intent signal parsed from the caller's own output stream triggers prefetch during generation, partial k* residency reduces switching to incremental layers, and a preemptive transfer scheduler protects live-session KV with a guaranteed bandwidth share. A carve-out control step adjusts the weight/KV allocation as workload demands change.

Evaluations show that JointServe reduces session completion time by 21.7\% over the
strongest whole-model swap baseline and stays usable in a crash regime where
every whole-model swap baseline fails to finish (a 46$\times$ gap). On
replayed SWE-bench, $\tau$-bench, and OSWorld traces, it cuts session
completion time by up to 49\% and lifts SLO attainment from
2.8--24.5\% to 76.8--100\%.
\end{abstract}

%% file: intro.tex
\section{Introduction}
\label{sec:intro}

Large language models are moving beyond single-turn question
answering toward multi-agent systems that plan, act, and revise over
long horizons~\cite{arxiv221003629,arxiv230803688,arxiv230713854,arxiv231112983}. Agents resolve repository issues
end to end~\cite{yang2024swe, liu2024large,arxiv240716741}, and multi-agent systems distribute
planning, coding, and reviewing roles across collaborating
models~\cite{hong2024metagpt,qian2024chatdev,tran2025multi,arxiv230808155,arxiv230317760}. A recurring architecture has emerged
from these systems: a general-purpose \emph{planner} model
orchestrates a pool of domain-specialized \emph{expert} models~\cite{arxiv230317580},
invoking an expert much as earlier agents invoked an external tool~\cite{arxiv230204761,arxiv230515334,arxiv240201030}.
The difference is where the capability lives. A tool exposes an
interface; an expert's coding, mathematical, or multimodal competence
is baked into its \emph{weights}, and no prompt can elicit a
capability from a model that is not loaded. Serving these workloads
therefore means serving many models on demand, not one model well.

Serving an expert pool requires memory for both model weights and
the state of active agent sessions. Private, edge, and shared-GPU deployments typically
offer one or two 80\,GB-class accelerators, while the combined
weights of even a modest expert pool overflow a single GPU.
Session state compounds the pressure: as an agent task deepens, the
key--value (KV) cache of each session grows with every turn, so
weights and KV blocks contend, physically and unavoidably, for the
same bytes of HBM. The cost of this contention lands on two
user-visible metrics that we use throughout the paper: time to first
token (TTFT), which governs the interactivity of each model call, and
session completion time (SCT), the wall-clock duration of an entire
agent task. The core challenge is not raw capacity; it is deciding
\emph{which bytes should reside} in scarce memory, and \emph{when
they should move}.

\begin{figure}[t]
  \centering
  \includegraphics[width=\columnwidth]{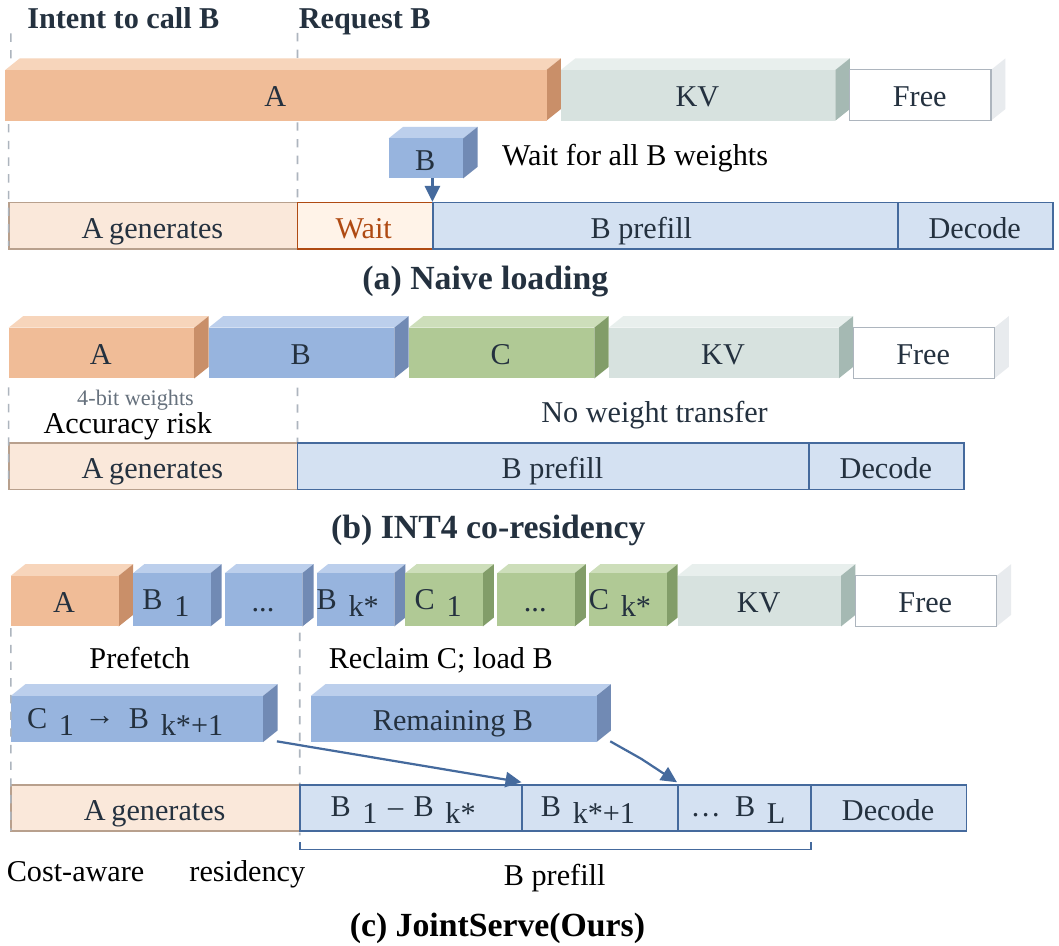}
  \caption{(a)~Naive loading: the request waits for the model weights to be transferred sequentially before prefill begins. (b)~INT4 co-residency: all quantized model weights remain simultaneously resident in GPU memory. (c)~JointServe: intent-guided shard prefetch before request arrival eliminates the waiting phase.}
  \label{fig:hero}
\end{figure}

What happens when the planner calls an expert that is not resident?
In a naive whole-model swap path, the serving stack must first make room in HBM, and then the entire
weight footprint of the target model streams in before prefill can
begin. This is a \emph{load-all barrier}: until the last layer is in
place, the first token cannot be produced
(Figure~\ref{fig:hero}(a)). Loading is not the only tax. When the
session later switches back to a model whose KV blocks were evicted
to make room without a host copy, that state is gone, and the context must be
re-prefilled from scratch. A single cross-model call therefore pays
twice: once in serial weight transfer on the way out, and once in KV
recompute on the way back. The first tax alone reaches second scale
on current hardware and grows linearly with model size, with the GPU
idle throughout; we quantify both taxes precisely in
\S\ref{sec:background}. These costs characterize the naive path in Figure~\ref{fig:hero}(a):
weights move at whole-model granularity, and evicted KV is not
preserved for fast restoration.

Prior work addresses different parts of this trade-off. \emph{Whole-model loading and session-aware
placement} reduce switch and return costs.
ServerlessLLM~\cite{fu2024serverlessllm} accelerates checkpoint
loading. Talaria~\cite{meng2026talaria} routes continuations using
model residency, KV locality, and instance pressure; soft reservations,
host-restorable KV, and staged weights preserve session continuity
across model switches. Talaria's cold pool still switches at model
granularity. JointServe instead selects resident weight layer groups
and KV blocks within one constrained GPU and prefetches missing groups
before a named call.
\emph{Quantized co-residency}
(Figure~\ref{fig:hero}(b)) removes swaps by shrinking weights~\cite{arxiv221110438,arxiv220807339,arxiv230603078,arxiv230607629,arxiv230813137}, but
merely relocates the contention: the bytes freed by quantization are
reclaimed by growing KV, the scheme fails outright when even
compressed weights do not fit, and it carries accuracy risk.
\emph{Same-family compression} shares a base model across adapters or
deltas~\cite{sheng2023s,yao2025deltazip,arxiv210609685,arxiv230514314,arxiv231018547}, an assumption that breaks in
heterogeneous expert ecosystems whose architectures, vocabularies,
and tokenizers differ. \emph{Passive caching}, from LRU eviction to
prefix reuse~\cite{zheng2024sglang}, decides what to keep from past reference
statistics alone; blind to which model comes next, it discards the
call-graph priors that agent workloads expose for free. These approaches improve different scheduling and memory costs, but
leave open how one memory-constrained GPU should jointly select weight
shards and session KV as demand changes.

This paper tackles exactly one question: under a single GPU's tight
memory budget, \emph{which bytes, weight shards or KV blocks, should
reside, when should they move, and on whose signal?} Our goal is to
have both: the capacity headroom of aggressive swapping and the
interactivity of full residency, without paying for one with the
other.

Talaria uses session-aware placement and admission across serving
instances~\cite{meng2026talaria}. JointServe instead manages weight
shards and KV blocks within one constrained GPU, using the caller's
named target to move bytes ahead of demand. Two
insights make this possible. \emph{First, residency is one problem,
not two.} Weight shards and KV blocks compete for the same memory
pool, so we model both as residency objects under a single budget
with explicit restore costs, and show that a fixed partition can become infeasible under the
measured crash budget. \emph{Second, agent call structure supplies
signals that passive statistics cannot.} A planner reveals its target
model while it is still generating the tool call, tool execution
leaves idle gaps between calls, and 96.4--99.9\% prefix reuse makes
preserved KV state highly valuable. Together
these observations turn the question of \emph{when} to move from a
prediction problem into a read-off problem.

JointServe occupies a point that complements, rather than competes
with, recent serving systems. BlitzScale~\cite{zhang2025blitzscale} scales
model instances at the cluster level; Pensieve, HCache, and
MORI~\cite{yu2025stateful,gao2025fast,xia2026idleness} tier the KV state of a single model
across the memory hierarchy; NanoFlow~\cite{zhu2025nanoflow} overlaps
compute, memory, and communication within an operator pipeline; and
DeltaZip~\cite{yao2025deltazip} compresses families of same-base
fine-tunes. JointServe manages the \emph{unified residency of
heterogeneous weights and KV within a single instance}.

Concretely, this paper makes three contributions.

\begin{itemize}
  \item \textbf{Unified residency formulation and adaptive
  split.} We formulate
  weight shards and KV blocks as same-domain residents of a single
  memory budget under an explicit cost model
  $c_i = \alpha + s/\mathrm{BW} + r(n)$, use a carve-out step to
  adapt the split during serving, and prove a $(1-\varepsilon)$
  recovery-value approximation for one selection snapshot of
  independent, additive objects.
  \item \textbf{Caller-generation proactive prefetching.} We extract agent call-graph semantics into
  residency signals. The planner's streaming tool-call header
  deterministically identifies the target model, moving the swap
  window from after request arrival to the caller's generation phase
  and tool gaps; a preemptive, chunk-granular transfer scheduler with
  a $\kappa$-floor bandwidth share keeps weight movement from
  starving KV restoration of live sessions.
  \item \textbf{Partial-shard streaming hot-swap and dual-state KV
  tiering.} Each of
  $N$ heterogeneous models keeps an analytically derived resident
  shard of $k^*$ layers with per-group readiness checks, so a swap
  streams only the remaining $L - k^*$ layers; the evaluated prototype
  switches by in-place overwrite, and a dual-state KV tier substitutes
  host-RAM reload for GPU recompute at a 20$\times$ cost advantage.
\end{itemize}

We implement JointServe on top of SGLang~\cite{zheng2024sglang} and
evaluate it across a range of real-world traces and datasets.
Compared with the strongest whole-model swap baseline, JointServe reduces
SCT by 21.7\% in the tight memory regime, stays feasible in a crash
regime where every whole-model swap baseline fails to finish (a
46$\times$ gap), and on real agent traces lifts SLO attainment from
2.8--24.5\% to 76.8--100\% while cutting SCT by up to 49\%.

%% file: background.tex
\section{Background and Motivation}
\label{sec:background}

\subsection{Model-as-Tool Agents and Serving Preliminaries}
\label{sec:bg-prelim}

Multi-agent systems increasingly follow a \emph{model-as-tool}
architecture: a general-purpose \emph{planner} model decomposes a
task, streams intermediate reasoning, and invokes domain
\emph{experts} through tool calls, which execute on the tool's output
and resume planning. Software-engineering agents call a coding model
behind a fixed agent--computer interface~\cite{yang2024swe}, and
multi-agent systems cast planning, implementation, and review as
roles played by different models~\cite{hong2024metagpt,qian2024chatdev}. The expert
pool assembled this way is genuinely heterogeneous, for three
recurring reasons. Domain experts are fine-tuned from different base
models and released independently; multimodal experts, such as vision
or code models, carry their own architectures, vocabularies, and
tokenizers; and legacy models remain in service long after their
training stacks are retired. Expert heterogeneity matters to a serving
system: it rules out sharing a single base weight across the pool,
which precludes adapter- and delta-based co-residency schemes. Recent agent-serving systems schedule
program structure and call dependencies well, but they assume that
every model an agent may invoke stays resident at some serving
endpoint~\cite{luo2025autellix,lin2024parrot}. JointServe starts where that
assumption breaks.

Serving engines such as vLLM and SGLang process each request in two
phases~\cite{kwon2023efficient,zheng2024sglang}. \emph{Prefill} consumes the prompt in one
forward pass and materializes its key--value (KV) state, which the
engine caches in HBM as paged blocks; \emph{decode} then generates
tokens autoregressively, reusing that state at every step. Unlike request-level latency models in the serving
literature~\cite{arxiv230505920,arxiv240302310,arxiv240109670,arxiv231118677}, which separate queuing from execution, a cross-model call
must also pay for memory movement. A
cross-model call traverses four latency components,
\begin{equation}
  T = T_{\text{queue}} + T_{\text{swap}} + T_{\text{restore}} + T_{\text{gen}},
  \label{eq:latency}
\end{equation}
where $T_{\text{queue}}$ is queuing delay; $T_{\text{swap}}$ is the
time to make room in HBM and transfer the target model's weights,
which delays the first token and therefore sets TTFT;
$T_{\text{restore}}$ is the time to recover the caller session's KV
state after a switch back, either by reloading preserved blocks from
host memory or by re-prefilling them on the GPU if they were dropped; $T_{\text{gen}}$ is the target
model's own execution, its prefill and decode, once the weights are
resident. TTFT therefore spans $T_{\text{swap}}$ up to the start of
target prefill. In \S\ref{sec:evaluation} we attribute observed
stalls to weight transfer, KV reload, and KV recompute under this
breakdown.

Weights and KV blocks ultimately contend for HBM. HBM sits on the
package and feeds tensor cores at full bandwidth, but its capacity is
fixed per device; host RAM and NVMe are larger by orders of magnitude,
yet narrower and slower by the same factor~\cite{arxiv230306865,arxiv240700079}. Only HBM can sustain
decode throughput, so weights and KV blocks must contend for its
constrained space. On an A800-80GB GPU over PCIe 4.0, the effective
host-to-device bandwidth is 24\,GB/s, so swapping in a 7B FP16 expert
(15.2\,GB) costs at least 0.63\,s with the GPU producing nothing. The
capacity arithmetic closes the argument: five experts spanning 7B to
14B parameters carry a combined FP16 weight footprint of 70--140\,GB,
which exceeds the 80\,GB of one card and approaches the budget of two.
Keeping the entire expert pool resident is therefore not a tuning
choice but arithmetic: under the one- or two-accelerator budgets of
private, edge, and shared-GPU deployments, full residency is out of
reach. One more card defers the wall rather than removing it, because
the pool grows with the number of experts served. Testbed and
measurement details appear in \S\ref{sec:evaluation}.

\subsection{Why Not Swap Whole Models Faster?}
\label{sec:bg-swap}

The most direct response to capacity pressure is to reduce switching
cost. ServerlessLLM~\cite{fu2024serverlessllm} accelerates checkpoint
loading, while Talaria~\cite{meng2026talaria} stages model weights
across cold-pool switches and preserves host-restorable session KV.
These methods reduce exposed transfer and recovery time. A conventional
whole-model reload, however, admits prefill only after the target
weights are ready. When the weights have not been staged, moving the
whole target model remains on the first-token path. JointServe instead
keeps a startup prefix resident and streams later layer groups as
execution proceeds.
Figure~\ref{fig:swap-transfer} quantifies the tax. Swapping in a 7B
FP16 expert (15.2\,GB) at the full 24\,GB/s host-to-device bandwidth
costs 0.63\,s with the GPU producing nothing, and the serial swap plus
prefill reaches 0.74\,s of TTFT; both grow linearly with model size. The default
engine reload path, which stops the instance, streams the full
checkpoint off NVMe into host memory, copies it over PCIe into HBM,
and re-initializes the runtime before serving resumes, is two orders
of magnitude worse at 67.76\,s for a 7B model: it serializes disk
reads, host-to-device copies, and engine re-initialization,
overlapping none of them with compute and serving no layer before the
last one lands.

\begin{figure}[t]
  \centering
  \includegraphics[width=\columnwidth]{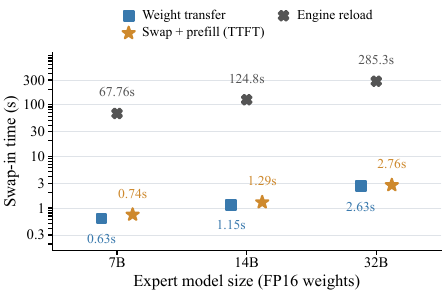}
  \caption{Overhead of reloading model weights.}
  \label{fig:swap-transfer}
\end{figure}

Faster transfer removes only the outbound tax, the serial weight
transfer before the switch; the return tax is untouched. The swap path
neither notices nor protects in-flight KV state, so blocks evicted to
make room are gone when the session switches back, and the full
context must be re-prefilled.
Figure~\ref{fig:kv-recompute} shows that this recompute cost grows quadratically with context length. At the median
$\tau$-bench context of 2.7K tokens it is still modest (0.18\,s), but
at the median SWE-bench context of 27.5K tokens one switch-back costs
2.64\,s, roughly four whole-model transfer lower bounds, and at the
p95 context of 61.8K tokens it reaches 8.89\,s. Restoring the same
state from host memory instead costs 0.05\,s, nearly independent of
length.

\begin{figure}[t]
  \centering
  \includegraphics[width=\columnwidth]{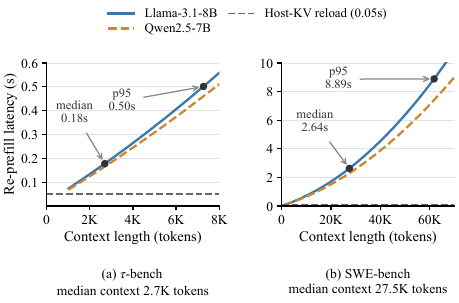}
  \caption{Re-prefill latency after KV-cache eviction as a function of evicted context length.}
  \label{fig:kv-recompute}
\end{figure}

The final limit is capacity, and Figure~\ref{fig:capacity-wall}
shows the wall forming well before the largest model. The figure
scales both JointServe and the whole-model route across model sizes on
a 72\,GB effective budget. At 14B the whole-model route still fits and
JointServe cuts SCT p50 by 65\% (978\,s to 343\,s); at 32B the margin
narrows to 23\% (1962\,s to 1514\,s) as the co-residency window
shrinks; at 70B, whose FP16 weights total 140\,GB, the whole-model
route is not slow but \emph{infeasible}, while a resident shard of
$k^* = 23$ layers (42\,GB) keeps JointServe serving. And the wall only
closes in as models grow, because weight footprints scale with width
and depth while the card's budget does not.
The verdict is that swapping faster is not swapping less,
and swapping has a capacity ceiling: the route pays both transfer and
recompute, link optimization touches only the first, and once weights
exceed the budget of a card the route ceases to exist. End-to-end
evidence that tight budgets still leave it 21.7\% behind, and that
collapse-budget workloads stall it entirely, appears in
\S\ref{sec:evaluation}.

\begin{figure}[t]
  \centering
  \includegraphics[width=\columnwidth]{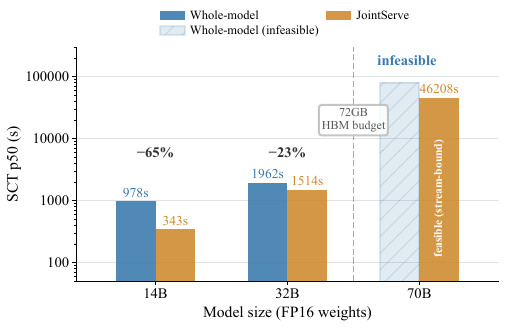}
  \caption{SCT p50 of whole-model swapping and JointServe across model sizes (72\,GB HBM budget).}
  \label{fig:capacity-wall}
\end{figure}

\smallskip
\noindent\textbf{O1: Calls are directed, not predicted.}\label{sec:bg-opp}
The streamed tool-call header identifies the next expert before request arrival, allowing prefetch during caller generation; 54\% of SWE-bench tool gaps precede a switch to a different tool.

\subsection{Why Not Shrink Models to Co-Reside?}
\label{sec:bg-quant}

A second family of designs avoids swapping altogether by making the
weights small enough that the whole pool fits at once: quantize every
expert to four bits, or share one base model across many
adapters~\cite{lin2024awq,frantar2022gptq,sheng2023s,yao2025deltazip}. Shrinking the weights,
however, merely relocates the contention rather than removing it.
Figure~\ref{fig:quant}(a) makes the arithmetic concrete at a 32\,GB
effective budget. A whole-model swapper locks 15.2\,GB for one FP16
expert and leaves 16.8\,GB for KV; the INT4 co-resident scheme locks
21.1\,GB for \emph{all five} experts and leaves only 10.9\,GB for KV.
Compression has pinned more bytes in weights that may sit idle for
whole stretches of a session, and the space left to the state that
grows without bound shrinks. Nor does accuracy come for free:
low-bit quantization degrades mathematical and code competence, a
cost documented in the quantization literature~\cite{lin2024awq,frantar2022gptq}.
Adapter- and delta-based sharing is bounded by an even earlier
constraint: it presumes a shared base model, which the heterogeneous
expert pool does not have.

\begin{figure}[t]
  \centering
  \includegraphics[width=\columnwidth]{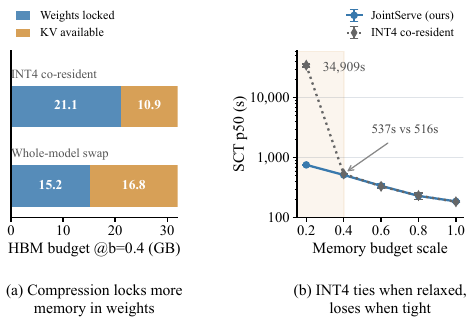}
  \caption{Shrinking weights relocates the contention.}
  \label{fig:quant}
\end{figure}

Figure~\ref{fig:quant}(b) compares INT4 co-residency against
JointServe across effective memory budgets, expressed as a fraction
$b$ of the 80\,GB card. When the budget is generous both coincide,
because the contention the scheme was built to avoid never
materializes; at a 32\,GB budget ($b = 0.4$) the co-resident scheme is
slightly worse (537\,s versus 516\,s), having traded KV headroom for
resident weights without buying anything. At a 16\,GB budget
($b = 0.2$) it fails outright, and the failure is structural: the
five-expert INT4 pool needs 21.1\,GB, which no longer fits, so
co-residency degenerates into swap-on-demand, the very mechanism it
was built to replace, and makespan collapses to 34{,}909\,s, within
5\% of the naive whole-model baseline, while JointServe completes the
same workload in 758\,s. Compression subtracts from one side of the
weight--KV contest; it never cancels it.

\smallskip
\noindent\textbf{O2: Execution is gapped.}
In the SWE-bench gap distribution shown in Figure~\ref{fig:opp}, a 0.14\,s incremental transfer fits 98\% of gaps versus 28\% for a 0.63\,s whole-model transfer.

\begin{figure}[t]
  \centering
  \includegraphics[width=\columnwidth]{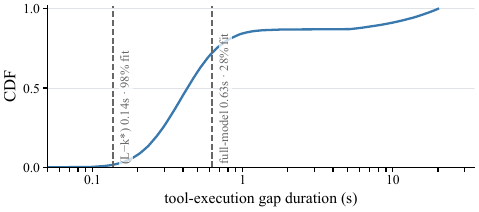}
  \caption{Cumulative distribution of tool-execution gaps in SWE-bench. Vertical lines mark the times required to transfer the incremental $(L-k^*)$ layers and the whole model.}
  \label{fig:opp}
\end{figure}

\subsection{How Does Session-Aware Routing Differ?}
\label{sec:bg-route}

Talaria~\cite{meng2026talaria} treats agent continuations as stateful.
Its router weighs model residency, KV locality, and instance pressure;
soft reservations preserve likely return affinity, while instance-local
host KV can be restored across model switches. Cold-pool instances
time-share multiple models and stage weights, so Talaria neither
requires a full replica of every model nor discards session KV at each
switch. Its main decisions concern where and when a session runs in a
serverless pool. JointServe addresses a complementary decision within
one memory-constrained GPU: which layer-group weight shards and session
KV blocks share HBM, and which missing groups to prefetch using the
caller's intent. The two policies act at different scheduling levels;
their integration is outside this paper's evaluation.

Whole-model staging, compression, and session-aware routing address
transfer, weight footprint, and cross-instance continuity, respectively;
JointServe instead selects weight shards and KV blocks jointly within
one GPU and combines this with intent-driven prefetch and partial
residency (\S\ref{sec:overview}).

\smallskip
\noindent\textbf{O3: Eviction is not destruction.}
Across four traces, prefix reuse is 96--100\%; at the SWE-bench median context, host KV reload takes 0.05\,s versus 2.64\,s for re-prefill, making host RAM a useful second residency tier.

%% file: overview.tex
\section{Unified Residency Management}
\label{sec:overview}

JointServe selects weight shards and KV blocks under a shared HBM budget using agent intent, reuse estimates, and recovery costs. This section defines the objective, joint budget rule, and online selection procedure.

\begin{figure*}[t]
  \centering 
  \includegraphics[width=\textwidth]{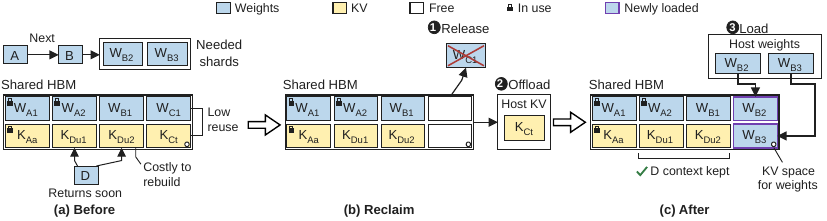}
  \caption{An illustration of joint residency selection and update.}
  \label{fig:joint-residency}
\end{figure*}

\subsection{Unified Residency Model}
\label{sec:model}

Figure~\ref{fig:joint-residency}(b) shows the candidates
$\mathcal{C}=\mathcal{W}\cup\mathcal{K}$: independently managed weight
layer groups ($\mathcal{W}$) and session KV units with known recovery
sources ($\mathcal{K}$). Each retains its model or session identity,
so KV for one model cannot serve another.

The budget $B$ excludes runtime, activation, and transfer buffers.
Protected objects $\mathcal{P}$ reserve $P_W$ weight bytes and $P_K$
KV bytes, leaving $B_f=B-P_W-P_K$ for selection. They include in-flight
session state and the arriving model's committed streaming prefix;
selection cannot evict them. If $\mathcal{P}$ alone exceeds $B$, the
execution layer handles the shortfall.

Candidate $i$ has size $s_i$, reuse probability $p_i$ within the
decision window, and cheapest-path recovery cost $c_i$. Its retention
value $u_i=p_ic_i$ is expected avoided recovery work, and its density
is $d_i=u_i/s_i$. The window counts one avoided recovery per object;
for repeated reuses we use expected counts.

The residency problem maximizes
\begin{equation}
  \max_{S \subseteq \mathcal{C},\; s(S) \le B_f} \; U(S),
  \qquad U(S) = \textstyle\sum_{i \in S} u_i ,
  \label{eq:objective}
\end{equation}
where $s(S)=\sum_{i\in S}s_i$. JointServe solves this objective using
estimated values $\hat{u}_i=\hat{p}_i\hat{c}_i$.
Appendix~\ref{app:oracle} bounds the effect of estimation error
on the selected set's recovery benefit. Translating that benefit
into lower TTFT and SCT also requires timely transfers and overlap
with computation (Appendix~\ref{app:latency}).

Equation~\eqref{eq:objective} assumes independent candidates and
additive values, which \S\ref{sec:selection} optimizes. Shared-prefix
recovery or mutually exclusive legal prefix and layer-group states
instead require complete legal plans per model or session.

\subsection{Eviction and Recovery Costs}
\label{sec:costs}

Weight objects reload from host. KV objects reload if a valid host copy
exists; otherwise they recompute the missing range. Eviction may also
require write-back, whose source buffer remains occupied until the
transfer finishes. These recovery paths enter the cost model.

Reload costs combine a transfer startup overhead with the time to
move the object's bytes. Recompute costs depend on the missing
token range and the model's prefill curve:
\begin{equation}
  \begin{gathered}
  c_i^{W} = \alpha_W + s_i/\mathrm{BW}_{H2D}, \\
  c_i^{K,\mathrm{load}} = \alpha_K + s_i/\mathrm{BW}_{H2D}, \\
  c_i^{K,\mathrm{recompute}} = \Delta r_m(n_i^{\mathrm{ready}},n_i),
  \end{gathered}
  \label{eq:costs}
\end{equation}
where $\alpha_W$ and $\alpha_K$ are transfer startup costs and
$\mathrm{BW}_{H2D}$ is effective host-to-device bandwidth.
For a context of $n_i$ tokens with a usable prefix of
$n_i^{\mathrm{ready}}$ tokens, $\Delta r_m$ is the recovery cost of
the missing suffix. We model a full-prefix recovery as
$r_m(n)=a_{0,m}+a_mn+b_mn^2$ and charge a partially missing unit
only the increment over its range. Shared recovery work, including
startup, is counted once. Startup overhead and quadratic prefill
cost can make recovery value per byte non-monotonic in prefix
length; Appendix~\ref{app:cost-density} derives this relationship.

\begin{figure*}[t]
  \centering
  \includegraphics[width=\textwidth]{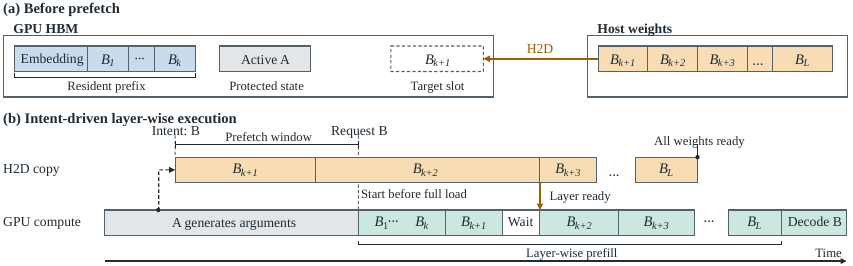}
  \caption{Intent-driven partial residency and layer-wise execution.}
  \label{fig:intent-streaming}
\end{figure*}

Equal-size objects can have different retention values: a
recompute-bound KV block may outweigh a host-reloadable block or weight
shard because prefill cost grows with missing context
(Figure~\ref{fig:kv-recompute}). JointServe ranks them by estimated
avoided recovery per byte, $\hat p_i\hat c_i/s_i$, using the reuse
signals of \S\ref{sec:signals} and the costs above.

\subsection{Joint Budget Allocation}
\label{sec:budget}

JointServe selects objects under the shared capacity rather than
fixing a weight--KV split. For each slow-path candidate, it adds the
best use of residual capacity; this decomposition is exact for
additive values and needs no continuity or equal-marginal assumption
(Appendix~\ref{app:decomposition}).

The slow path forms a frontier of representatives $(R_q,s_q,v_q)$
from weights and compressed high-value KV: each stores a retained set,
its true size, and its estimated value. For each rounded-value state,
it keeps the smallest set, including the empty set; value scaling
bounds the state count while sizes stay exact (Appendix~B.3). The fast
path queries $F_L(r)$ for the low-value KV set fitting residual
capacity $r$.

Selection maximizes over the frontier,
\begin{equation}
  S^* = \arg\max_{R_q}\; \hat{U}(R_q) +
        \hat{U}\!\left(F_L(B_f - s_q)\right),
  \label{eq:joint}
\end{equation}
using true sizes and same-window values. High-value KV remains KV bytes. The winner sets
$b_W=P_W+s(S^*\cap\mathcal{W})$ and $b_K=B-b_W$; the KV quota includes
protected KV, selected blocks, and slack. These quotas are derived
afresh when intent, estimates, protection, or capacity change. Each
event rebuilds the frontier and KV index, then checks protection
feasibility and the approximation bound.

\begin{algorithm}[t]
  \caption{Joint Residency Selection}
  \label{alg:joint}
  \begin{algorithmic}[1]
    \Require candidates $\mathcal{W}, \mathcal{K}$; protected set
    $\mathcal{P}$; free bytes $B_f$; scores $\hat{u}$; budgets
    $\varepsilon_C, \varepsilon_S, \varepsilon_K > 0$,
    $\varepsilon_C + \varepsilon_S + \varepsilon_K < 1$
    \State \textbf{if} protection $\mathcal{P}$ is infeasible
    \textbf{then return} \textsc{CapacityShortfall}
    \State drop individually infeasible or nonpositive-value candidates
    \State \textbf{if} $B_f \le 0$ or no candidates remain
    \textbf{then return} target $\mathcal{P}$
    \State $G \gets$ value of the feasible density prefix over
    $\mathcal{W} \cup \mathcal{K}$; $M \gets$ best feasible singleton;
    $L \gets \max(G, M)$
    \State split $\mathcal{K}$ at $\tau = \varepsilon_K L$ into
    $\mathcal{K}_H$ and $\mathcal{K}_L$
    \State $\mathcal{K}_H' \gets
    \textsc{CompressByValueAndSize}(\mathcal{K}_H, \tau, \varepsilon_C,
    \varepsilon_K)$ \Comment{Appendix~\ref{app:construction}}
    \State $\mathcal{R} \gets
    \textsc{MinBytesValueFrontier}(\mathcal{W} \cup \mathcal{K}_H',
    B_f, \varepsilon_S)$ \Comment{Appendix~\ref{app:frontier}}
    \State $F_L \gets \textsc{DensityPrefixIndex}(\mathcal{K}_L)$
    \State $\mathit{best} \gets \emptyset$
    \ForAll{representatives $R_q \in \mathcal{R}$, including $\emptyset$}
      \State $K_q \gets F_L.\texttt{query}(B_f - s(R_q))$
      \If{$\hat{U}(R_q \cup K_q) > \hat{U}(\mathit{best})$}
        \State $\mathit{best} \gets R_q \cup K_q$
      \EndIf
    \EndFor
    \State target $\gets \mathcal{P} \cup \mathit{best}$; derive $b_W,
    b_K$ by object type
    \State \Return target with its value and budget metadata
  \end{algorithmic}
\end{algorithm}

\subsection{Cost-Aware Residency Selection}
\label{sec:selection}

Algorithm~\ref{alg:joint} sets $G$ to the value of the feasible
density prefix ending before the first nonfitting candidate and $M$
to the best feasible singleton. For $L=\max(G,M)$,
$L\le\widehat{\mathrm{OPT}}\le2L$ (Appendix~\ref{app:scale}),
which sets the compression scale. Infeasible protection returns a
shortfall; if $B_f\le0$ or no positive-value candidate fits, the
procedure returns only $\mathcal P$.

The threshold $\tau=\varepsilon_KL$ separates high-value KV
($\hat u_i>\tau$) from low-value KV. In each geometric value bucket,
JointServe keeps the $h=\lceil2/\varepsilon_K\rceil$ smallest
high-value objects. Any optimum has fewer than $2/\varepsilon_K$
such objects, so these representatives preserve a feasible replacement
with at most $\varepsilon_C$ loss in that subset's value; Supplementary
Material, \S A.8 bounds the candidate count. Compression changes the
search space, not physical residency. Dependencies instead require
complete legal options per model or session.

A value-scaling dynamic program stores the smallest representative at
each rounded value for weights and compressed $\mathcal K_H'$;
$\varepsilon_S$ controls rounding. The low-value KV index sorts by
density and uses cumulative sizes and values to return the longest
fitting prefix $F_L(r)$ by binary search. Filling later fitting
objects may add value, but the guarantee uses this prefix.
Equation~\eqref{eq:joint} selects the target; Appendix~C gives the
exact Top-$k$ case for equal-size KV and the grouped-plan dynamic
program.

In Figure~\ref{fig:joint-residency}, intent names model B and session
D updates its reuse estimate; selection retains A's protected state
and D's context while admitting B's shards.

The target induces loads, releases, and offloads relative to current
residency; unchanged objects stay. Write-back completes before its
source buffer is released. The runtime in \S\ref{sec:design} commits
feasible transfers and reports capacity shortfalls without replacing
protected objects.

For a fixed set of independently selectable candidates with
nonnegative additive values and one capacity constraint,
Algorithm~\ref{alg:joint} returns a feasible set $S_A$ satisfying
$\hat U(S_A)\ge(1-\varepsilon)\widehat{\mathrm{OPT}}$, where
$\varepsilon=\varepsilon_C+\varepsilon_S+\varepsilon_K<1$.
Representatives use no more space than the sets they replace, and
compression, rounding, and prefix selection each consume their
assigned error budget. Appendix~\ref{app:proof} gives the proof;
Appendix~\ref{app:estimation} relates the estimated objective to
true expected recovery benefit.

%% file: design.tex
\section{JointServe System Design}
\label{sec:design}

\begin{figure*}[t]
  \centering
  \includegraphics[width=\textwidth]{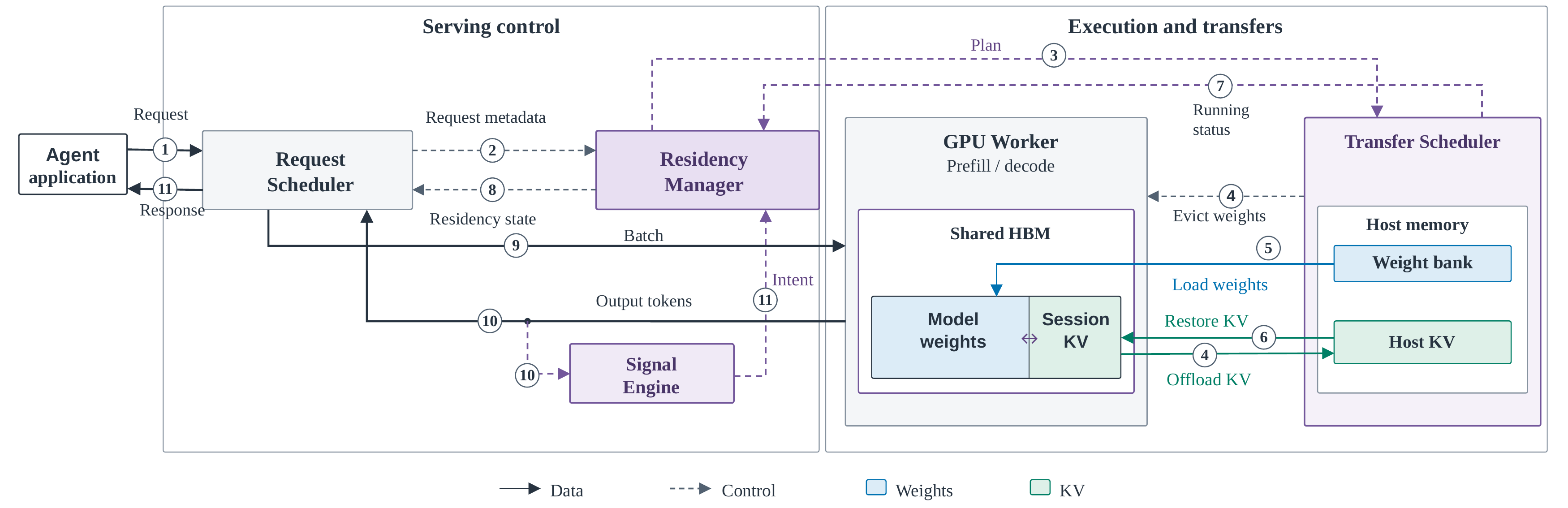}
  \caption{JointServe system overview.}
  \label{fig:architecture}
\end{figure*}

JointServe turns the residency selection of \S\ref{sec:overview}
into transfers and execution through the architecture in
Figure~\ref{fig:architecture}. The Request Scheduler, Residency
Manager, and Signal Engine form the serving-control layer, which
produces versioned residency plans. The GPU Worker and Transfer
Scheduler execute those plans across shared HBM and host memory.
This separation lets the residency policy operate independently of
the loading back end. 

Executing a residency plan requires making room for the next model while preserving the state of in-flight sessions. The Residency Manager tracks the selected target separately from the objects ready for execution. Within the shared HBM budget, the Transfer Scheduler overlaps weight loading with KV movement where dependencies permit and reserves bandwidth for live-session KV traffic. The GPU Worker begins prefill once the startup layer groups and required session KV are resident, while later groups stream in alongside computation. Per-group ready events enforce the remaining dependencies, allowing execution to begin before the full model arrives.

\subsection{Signal-Driven Model Anticipation}
\label{sec:signals}

Agent call structure turns which model comes next from a prediction
problem into a read-out, and the Signal Engine's job is narrow: it
produces exactly the same-window inputs that
Algorithm~\ref{alg:joint} consumes. Signals follow a quality
ordering. An explicit intent is near-certain evidence, a
session-level online Markov transition is weaker, and popularity or LRU
recency is the fallback. The intent tap sits on the GPU Worker's
output stream toward the Request Scheduler
(Figure~\ref{fig:architecture}, step 11), a bypass that observes
without rerouting: a lightweight lexical matcher locks the target
model as soon as the caller streams the function-call header, and
fan-out headers generalize to a candidate set for next-set
prefetch. This realizes the advance of opportunity O1
(\S\ref{sec:bg-opp}) from request arrival into the caller's own
generation phase, and it opens the prefetch window of
Figure~\ref{fig:intent-streaming}.

Without an explicit intent, the engine falls back to a first-order
Markov model over call transitions, updated asynchronously with
Laplace smoothing as calls complete, and finally to popularity. A
misprediction costs bounded work, not correctness, because bytes
fetched for the wrong target are priced by the same value model and
evicted like any candidate, so the worst case degrades to the
passive baseline. For each candidate, the Signal Engine supplies a
reuse estimate $\hat{p}_i$ over the decision window defined in
\S\ref{sec:model}. The Residency Manager combines this estimate
with the candidate's recovery-path cost $\hat{c}_i$ to obtain
$\hat{u}_i=\hat{p}_i\hat{c}_i$ for Algorithm~\ref{alg:joint}.
Appendix~\ref{app:estimation} bounds how errors in these inputs
affect the selected set's expected recovery benefit.p.

\subsection{Executing Residency Plans}
\label{sec:plans}

A plan is a differential, not a re-derivation:
Algorithm~\ref{alg:joint} hands the Residency Manager a target
residency set, its per-type quotas, and a diff of actions, release
or offload displaced objects and load selected ones
(Figure~\ref{fig:architecture}, steps 3--6). Execution rests on one
identity and one registry: every object keeps its model or session
key end to end, and the host side registers a read-only weight bank
plus a host KV store keyed by model.

Commit is a three-step protocol: the delta is dispatched, with
evictions and offloads running in parallel with loads and restores
(steps 4--6). The Transfer Scheduler reports startup readiness once
the groups needed to begin execution are resident (step 7), and the
manager atomically switches the residency mapping and hands control
back to the Request Scheduler (steps 8--9). Target set and ready
set are tracked as distinct states: a new mapping is never
committed before reclamation finishes and the target copies are in
place, and the protected set of \S\ref{sec:model} is as unevictable
here as it was in selection. If a plan is infeasible or capacity
falls short, the old mapping is retained and the request returns to
the decision layer, the runtime counterpart of the
\textsc{CapacityShortfall} branch of Algorithm~\ref{alg:joint}.

The plan determines which objects move; the Transfer Scheduler
determines their order. Per-group ready
events from the runtime establish when the startup prefix and
required session state can be used.
Remaining groups must arrive by their first use; otherwise,
execution waits on the corresponding ready event. These execution checks preserve the
selected target while tracking which parts are ready.

\subsection{Preemptive Transfer Scheduler}
\label{sec:transfer}

Three traffic classes share one host-to-device channel: target-model
prefetch, the streaming shard refill of \S\ref{sec:streaming}, and
background KV offload and restore. Left unordered, a
gigabyte-scale weight burst would starve the KV restoration of
sessions already in service.

The scheduler enforces a static priority order, target prefetch
above pipeline shard stream above background KV, and chunks weight
transfers at layer-group granularity so a higher-priority transfer
preempts an in-flight one at chunk boundaries. A $\kappa$-floor,
15\% of link bandwidth by default, reserves a minimum share for the
KV traffic of live sessions: a lower bound on what KV restoration
receives, not an upper bound on what weights may use when the
channel is otherwise idle. Per-task readiness reports aggregate
into the transfer status of step 7 in
Figure~\ref{fig:architecture}. The scheduler only orders the deltas
that \S\ref{sec:plans} hands down and never alters the target set.
Its priority input comes from the Signal Engine, whose intent marks
the highest level, so signal, decision, and transfer depend on one
another in a single direction.

\subsection{Sharded Runtime and Streaming Hot-Swap}
\label{sec:streaming}

Partial residency refines the load-all barrier of
\S\ref{sec:bg-swap} into per-group dependencies. The runtime keeps a
startup prefix of $k^*$ layer groups resident per model and checks
each candidate prefix against per-group load times, first-use
deadlines, and buffer lifetimes. Sufficient conditions under aligned
release times and one reserved channel appear in
Appendix~\ref{app:deadlines}. Execution then proceeds on
ready events: a group computes once its copy completes, and any
residual wait is exposed rather than assumed away.

Figure~\ref{fig:intent-streaming} follows one hot swap. While
caller A is still generating, its intent opens the prefetch window
(a): B's incremental layer groups move over the host-to-device
link, overlapped with A's decode, and the window closes when B's
request arrives. Prefill starts immediately on the resident layers
(b), every later group executes only after its own copy completes,
and the whole model is resident before decode begins. A
shadow-buffer ping-pong flip is the target hand-off between
prefixes, while the current prototype overwrites in place, a
difference \S\ref{sec:implementation} states concretely. The Load
edge of Figure~\ref{fig:joint-residency}(b) is this figure's layer
stream: the space and time views of the same objects.

Sessions recover along the same plan-driven paths: a switch back
reloads its KV from host memory (step 6 of
Figure~\ref{fig:architecture}), reload versus recompute is decided
online by the cost model of \S\ref{sec:costs}, and offload is
ordered before release, the execution-layer form of the staged
accounting in \S\ref{sec:costs}. Host capacity is managed under the
same value model, and the $\kappa$-floor of \S\ref{sec:transfer}
protects the restore traffic of sessions in service.

Layer-granular loading is an established lineage:
PipeSwitch~\cite{258864} pipelines one model's layers over PCIe
for context switching, ZeRO-Inference~\cite{10046087} streams
a single oversized model from CPU or NVMe memory, and
$\lambda$Scale~\cite{yu2026lambdascaleenablingfastscaling} and BlitzScale~\cite{zhang2025blitzscale}
amortize loading across a cluster. All four move the layers of one
model, or share one model's loading cost across instances. None
keeps several heterogeneous models partially resident in one memory
pool and reassembles the working set per call. The two directions
are orthogonal and compose: a single instance can run the
$k^*$-prefix scheme here while the cluster layer reuses a multicast
loading path.

%% file: implementation.tex
\section{Implementation}
\label{sec:implementation}

JointServe is implemented on SGLang~\cite{zheng2024sglang}: the system
core is about 1.5K lines of Python. Three hook groups form the integration surface.
\texttt{submit\_plan(target, delta)} hands the runtime one versioned
residency plan: a target set plus the release, offload, load, and
restore actions that reach it. \texttt{ready\_event} reports one
layer group's copy completion to the commit protocol.
\texttt{kv\_load} and \texttt{kv\_offload} move session KV between
the HBM and host tiers, and the intent tap observes the engine's
output stream without rerouting it.

The shard loader reads layer groups from a pinned, read-only host
weight bank and copies them host-to-device on a dedicated stream
that overlaps compute; page-locked staging keeps the copy at link
rate, and the per-group events aggregate into the startup-readiness
report that the commit protocol of \S\ref{sec:plans} waits on. Two
residency managers, one for weight shards and one for KV blocks,
hold the residency maps keyed by model and session. The stock reload
path re-reads a full checkpoint from disk and admits execution only
after the entire state dict is resident, so JointServe bypasses the
engine loader and writes weight buffers directly from host memory.

The Signal Engine, Residency Manager, and Transfer Scheduler run as
asynchronous threads that consume only the output-stream tap and
engine heartbeats, never the request critical path. Plan deltas
execute concurrently, and only the final mapping switch is
serialized. Priorities and the $\kappa$-floor are chunked queues
inside the transfer channel, and preemption happens at chunk
boundaries. Amortized control-plane overhead stays below 0.1\% of a
decode step (\S\ref{sec:evaluation}).

%% file: evaluation.tex
\section{Evaluation}
\label{sec:evaluation}

We evaluate the full JointServe implementation against five baselines
on synthetic and real agent workloads. Three
results summarize the chapter:

\begin{itemize}[leftmargin=1.2em,itemsep=1pt,topsep=2pt]
  \item Across memory budgets, JointServe cuts median session
  completion time (SCT) by 21.7\% against the strongest whole-model
  baseline in the tight regime, and remains feasible at a crash budget
  where every whole-model baseline stalls past the measurement window.
  \item On real agent traces, JointServe raises SLO attainment on
  $\tau$-bench from 24.5\% to 100\% and cuts SWE-bench SCT by 49.1\%.
  \item The mechanism itself is fast: a streaming hot-swap changes the
  resident expert in 1.83\,s where an engine reload takes 67.76\,s
  (37$\times$), and the control plane amortizes to under 0.1\% of a
  decode step.
\end{itemize}

\subsection{Experimental Setup}
\label{sec:eval-setup}

\textbf{Testbed.} A single A800-SXM4-80GB GPU on PCIe~4.0, backed by
1\,TB of host RAM and NVMe SSDs. Measured effective host-to-device
bandwidth is 24.3\,GB/s; per-transfer startup constants are
$\alpha_W{\approx}0$ for weights and $\alpha_K{=}0.707$\,ms for KV
blocks. Memory budgets $b$ are fractions of the 80\,GB
HBM; 8\,GB is reserved for the runtime, so $b{=}0.4$ leaves a 24\,GB
usable budget.

\textbf{Model catalogs.} Three size tiers turn model scale into an
independent variable. The standard pool holds five 7B-class experts
(Llama-3.1-8B as planner plus Qwen2.5-7B-Instruct, Qwen2.5-Coder,
Qwen2.5-Math, and Mistral-7B; 15.2\,GB per expert in FP16), so the
pool fits no single-GPU budget and the tight regime is the main
battleground. The 14B tier (28--30\,GB per expert) lets one expert
start to squeeze the budget; the 32B tier (64--66\,GB per expert)
approaches the whole card, where whole-model residency stops being
physically possible.

\textbf{Workloads.} Three load families. Synthetic sessions with
lognormal depths and lengths, parameterized from published agent-trace
characterizations~\cite{meng2026talaria}, drive the budget sweep. A released
coding-agent trace (TraceLab~\cite{zhu2026tracelab}, 665K turns) is replayed
as a production-shaped check. Three public agent benchmarks are
measured directly: SWE-bench~\cite{jimenez2024swe},
$\tau$-bench~\cite{taubench}, and OSWorld~\cite{xie2024osworld}.

\subsection{Baselines}
\label{sec:eval-baselines}

Our System uses the five baselines:

\textbf{vLLM}~\cite{kwon2023efficient} restarts the engine on every model switch,
the naive path. 

\textbf{SGLang}~\cite{zheng2024sglang} re-reads weights from
NVMe and hot-reloads them without a restart, with HiCache-style LRU
tiering for KV state; this is the default path of a mainstream
serving framework. 

\textbf{ServerlessLLM}~\cite{fu2024serverlessllm} is the
fastest published loading stack; we configure it at its official
23\,GB/s loading bandwidth and additionally compare against an
idealized host-resident variant that pays only the H2D transfer, which
serves as an optimistic upper bound for the whole-model family.

\textbf{INT4-CoResident} keeps 4-bit weights of all experts resident
and never swaps, the compression-only route.

\textbf{INT4-Swap} combines 4-bit weights with whole-model swap, the
naive combination of the two routes. We report JointServe in two forms: JointServe
(FP16, lossless) and JointServe-INT4 (sharding composed with
quantization).

\subsection{Evaluation Metrics}
\label{sec:eval-metrics}

\textbf{SCT} is the
primary metric: how long a user waits for a whole agent task.
\textbf{TTFT} captures interactivity at the request level, and
\textbf{SLO attainment} converts it to a production-readiness number:
the fraction of requests with TTFT under 5\,s, a threshold in line
with what interactive agent workloads tolerate. \textbf{Stall
decomposition} attributes per-request waiting to three sources,
weight-switch, KV-reload, and KV-recompute, following the latency
anatomy of a cross-model call. \textbf{Model-switch count} exposes
behavior directly: whether partial residency eliminates switching
altogether.

\begin{figure}[t]
  \centering
  \includegraphics[width=\columnwidth]{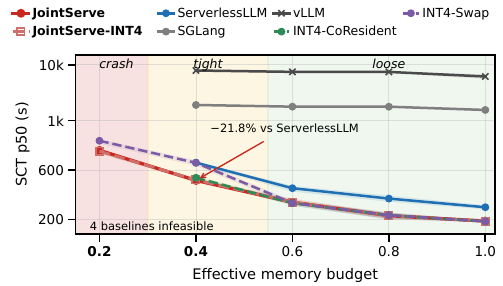}
  \caption{SCT p50 across memory budgets $b$ ($b$ is the fraction of one 80\,GB GPU's memory capacity).}
  \label{fig:budget-sct}
\end{figure}

\subsection{End-to-End Results}
\label{sec:eval-e2e}

\begin{figure*}[t]
  \centering
  \includegraphics[width=\textwidth,trim=0.32in 0 0.06in 0,clip]{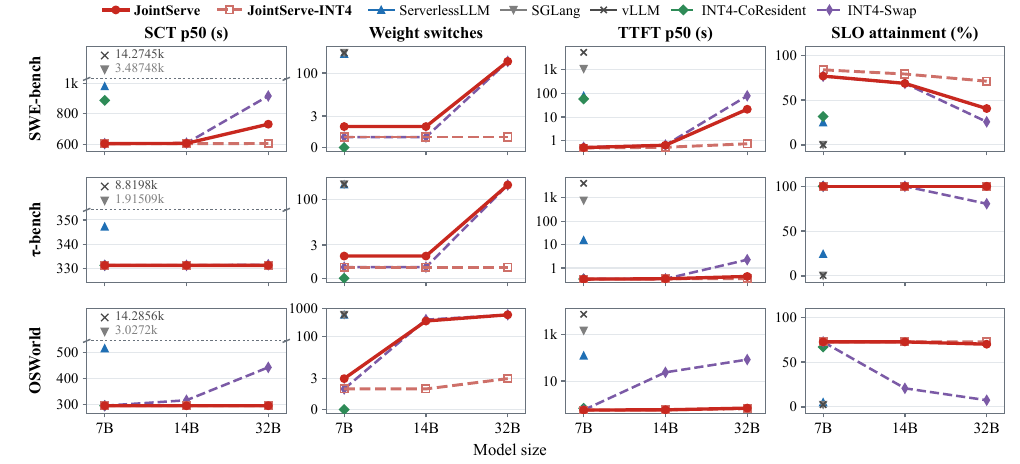}
  \caption{Real agent traces. JointServe
  wins on all three workloads.}
  \label{fig:agent-replay}
\end{figure*}

\textbf{JointServe dominates across all pressure regimes}
(Fig.~\ref{fig:budget-sct}). In the \emph{tight} regime ($b{=}0.4$,
24\,GB usable), JointServe completes the median session in 516\,s,
reducing SCT by 21.7--72.9\% over the stronger swapping baselines and
by 15$\times$ over vLLM. The advantage comes
from spending the scarce bytes on high-value objects rather than on
whichever whole model happens to be resident. In the \emph{crash}
regime ($b{=}0.2$, 8\,GB usable),
every baseline that must place a whole model does not finish within
the 35{,}000\,s measurement window; JointServe, which needs only a
$k^\ast$ shard per active model, completes in 758\,s, a 46$\times$
gap against the least-affected baseline's truncated reading. 

When capacity is ample ($b{=}1.0$), JointServe still reduces SCT
by 37.5\% relative to ServerlessLLM, while matching INT4-CoResident,
for which all weights already fit in HBM. This convergence is
expected: once residency ceases to be a constraint, there is little
room for sharding to improve capacity utilization. In contrast,
SGLang remains 8.3$\times$ slower than JointServe even at full HBM
capacity, indicating that its bottleneck lies in the storage-backed
reload path rather than in HBM capacity itself.

\textbf{Quantization composes with, and does not replace, sharding.}
At 7B under the tight budget, quantizing a whole-model swap reduces
transfer cost but retains the whole-model loading barrier;
JointServe-INT4 remains about 21\% faster than INT4-Swap. As the
budget increases, the two converge to within 2\%, since both
configurations can keep their working sets resident. The distinction
grows with model size: at 32B, JointServe-INT4 is 2.3$\times$ faster
than the best feasible alternative in the tight regime. Under the
crash budget, FP16 sharding and INT4 whole-model swapping both cross
the capacity limit, while JointServe-INT4 remains feasible at both
14B and 32B. Thus, quantization lowers the cost of movement, whereas
sharding removes the whole-model capacity requirement; combining the
two is the only evaluated configuration that remains robust as model
size grows.

\textbf{Real agent traces confirm the advantage outside the synthetic
sweep} (Fig.~\ref{fig:agent-replay}). On \emph{SWE-bench}, where
execution is dominated by model switching, JointServe reduces SCT by
49.1\% and raises SLO attainment from 2.8\% to 76.8\%, while reducing
physical model swaps from 478 to 2 per run. On \emph{$\tau$-bench},
the primary benefit shifts from completion time to responsiveness:
median TTFT improves by 45$\times$ and SLO attainment reaches 100\%,
whereas SCT improves by only 5.2\%. This difference is structural.
$\tau$-bench interleaves model calls with long user-thinking
intervals, allowing much of the baseline's swap latency to hide
inside otherwise idle periods; those swaps still delay the first
token of each response. On
\emph{OSWorld}, whose 16\,s-scale tool gaps are the friendliest
regime for prefetching, TTFT p50 falls from 31.7\,s to 0.1\,s, SLO
attainment goes from 2.8\% to 100\%, and SCT drops 43.3\%. Finally, replaying the TraceLab production
coding-agent trace reduces SCT by 51\%, confirming that the advantage
extends beyond the benchmark workloads.

Cost follows SCT: at \$2 per GPU-hour, serving
1{,}000 sessions costs \$10.30 with JointServe against \$14.67 for the
strongest baseline, a 30\% reduction from the same runs.

\begin{figure}[t]
  \centering
  \includegraphics[width=\columnwidth]{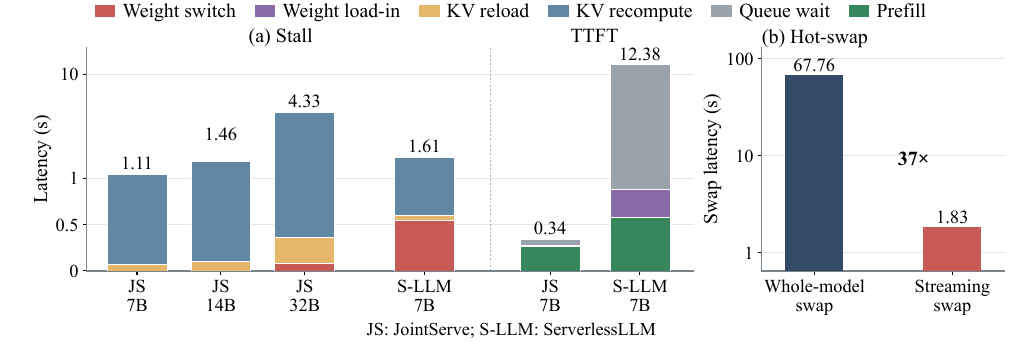}
  \caption{(a) Per-request stall breakdown at (b{=}0.4) (left) and TTFT breakdown (right). (b) Expert-swap behavior: JointServe overlaps layer streaming with prefill, while the default engine performs a blocking reload.}
  \label{fig:stall}
\end{figure}

\subsection{Breaking Down the Speedup}
\label{sec:eval-stall}

\textbf{Stall decomposition} (Fig.~\ref{fig:stall}a, left) shows where
JointServe's advantage is manufactured. At $b{=}0.4$,
ServerlessLLM pays 0.55\,s of weight-switch stall per request;
JointServe pays 0.00\,s at 7B and 14B and 0.11\,s at 32B, because
intent-driven prefetch moves the bytes before the request exists. On
the KV side, the dual-state tier replaces a 1.03\,s GPU recompute with
a 0.05\,s host reload, a 20$\times$ cost difference per switch-back.
When a request waits on several resources at once, we attribute the
wait once, in weight-switch, KV-reload, KV-recompute order.

\textbf{TTFT anatomy} (Fig.~\ref{fig:stall}a, right) decomposes the
$\tau$-bench result: JointServe's 0.343\,s TTFT consists of 0.007\,s
residency wait, 0.023\,s queue wait, and execution, while
ServerlessLLM's 12.385\,s is dominated by 11.497\,s of queue wait.
That gray segment bundles every way a request can be blocked by other
sessions, including their compute, their swaps, and their KV restores.
We do not claim it is all switching cost, but JointServe's 0.007\,s
residency wait is the hard floor on how much of it can be.
SWE-bench shows the same shape at the extreme: 137.0\,s of the
baseline's 137.5\,s median TTFT is waiting.

\textbf{Hot-swap microbenchmark and overhead}
(Fig.~\ref{fig:stall}b). On the A800, JointServe's streaming hot-swap
makes an expert usable in 1.83\,s (prefill starts at 0.11\,s; the
remaining 1.72\,s of layer copies overlap with compute), against
67.76\,s for the default engine path (43.31\,s NVMe read plus
24.45\,s H2D), a 37$\times$ reduction that translates into an 84\%
TTFT cut on swap-bound requests. The
control plane is effectively free: residency selection costs 811\,\textmu s
per 32-request batch, the carve-out step 0.5\,\textmu s per 64
requests, and signal extraction about 1.2\,\textmu s per request,
amortizing to 27.5\,\textmu s per request, under 0.1\% of a decode
step.

\subsection{Ablation}
\label{sec:eval-ablation}

\input{eval_ablation_table}

Table~\ref{tab:ablation} removes each pillar at $b=0.4$ on the standard workload, using the same code snapshot. Without \emph{partial residency}, SCT rises 18.0\% and weight switches increase from 98 to 257 because whole models must be transferred. Without the \emph{host KV tier}, SCT rises 53.8\% because switch-backs recompute on the GPU: KV recomputes increase from 106 to 286 and reloads fall to zero. Removing \emph{adaptive control} (intent prefetch, the carve-out budget, and $\kappa$-floor scheduling) raises SCT 24.3\% (516\,$\to$\,641\,s): transfers return to the request path and bulk movement delays in-flight KV restores.

At the crash budget ($b=0.2$), only 8\,GB remains usable after the runtime reserve, so no expert's full weights fit. Without partial residency, the run exceeds the 35{,}000\,s measurement window, as do the whole-model baselines. Without adaptive control, a fixed 50/50 split gives weights only 4\,GB, below the 4.5\,GB startup shard ($k^*$), and leaves no legal plan alongside in-flight KV. JointServe shifts capacity toward weights before a switch and back to KV during generation, finishing in 758\,s.

The controls have a super-additive effect: they are individually near-neutral on the steady workload ($-4.7\%$ for a fixed split, $+0.5\%$ without prefetch, and $+0.1\%$ without $\kappa$-floor), but removing them together raises SCT 24.3\% and crosses the feasibility boundary at $b=0.2$. The carve-out converges to its predicted 12.0\,GB equilibrium, and the residency selector's measured optimality gap averages 5.1\%.

%% file: eval_ablation_table.tex
\begin{table}[t]
  \centering
  \caption{Pillar ablation on the standard workload.}
  \label{tab:ablation}
  \setlength{\tabcolsep}{3pt}
  \footnotesize
  \begin{tabular}{@{}lrrrr@{}}
    \toprule
    Variant & SCT$_{p50}$ (s) & weight sw. & KV reloads & KV rec.\ \\
    \midrule
    JointServe (full) & 516 & 98 & 180 & 106 \\
    w/o Partial Residency & 608 (+18.0\%) & 257 & 187 & 105 \\
    w/o Host KV Tier & 793 (+53.8\%) & 98 & 0 & 286 \\
    w/o Adaptive Control & 641 (+24.3\%) & 115 & 189 & 103 \\
    \bottomrule
  \end{tabular}
\end{table}

%% file: discussion.tex
\section{Discussion}
\label{sec:discussion}

\textbf{Capacity determines gains and feasibility.} At loose budgets, compute saturates and residency policies converge; at $b{=}0.4$, JointServe saves 21.7\% of session time; at the crash budget, whole-model baselines do not finish within the measurement window. SGLang still has 8.3$\times$ JointServe's median SCT at full capacity because its storage-backed reload cost remains.

\textbf{Gap structure shifts the benefit.} Switch-bound SWE-bench mainly gains completion time, $\tau$-bench gains interactivity, and OSWorld's long tool gaps hide much of the transfer. Intent-driven prefetch reduces exposed stall, so end-to-end gains depend on the baseline's critical path.

\textbf{Adaptive allocation preserves feasibility.} The control mechanisms are super-additive: nearly neutral individually but decisive together. Shifting capacity between weights and KV fits working sets that a fixed split cannot hold at the crash budget. Steady-state speed alone understates this benefit.

%% file: related.tex
\section{Related Work}
\label{sec:related}

\emph{Model and request placement.}
ServerlessLLM~\cite{fu2024serverlessllm} speeds loading; BlitzScale~\cite{zhang2025blitzscale} loads layers for autoscaling. Talaria~\cite{meng2026talaria} routes by residency, KV locality, and pressure, preserving return continuity through soft reservations, mid-slot admission, host-restorable KV, and staged weights. MuxServe and Llumnix~\cite{arxiv240402015,arxiv240603243} multiplex or schedule instances. JointServe selects weight shards and KV blocks inside one GPU, retains startup prefixes, and streams missing groups.

\emph{KV and model families.}
Pensieve and HCache~\cite{yu2025stateful,gao2025fast} preserve or restore KV; MORI~\cite{xia2026idleness} moves it during idle windows. CacheBlend~\cite{arxiv240516444} reuses KV; CacheGen~\cite{arxiv231007240} compresses and streams it. InfiniGen, H$_2$O, StreamingLLM, SnapKV, and PyramidKV~\cite{arxiv240619707,arxiv230614048,arxiv230917453,arxiv240414469,arxiv240602069} manage long-context KV, and KIVI~\cite{arxiv240202750} quantizes it. AWQ and GPTQ~\cite{lin2024awq,frantar2022gptq} compress weights; S-LoRA and Punica~\cite{sheng2023s,arxiv231018547} serve shared-base adapters; DeltaZip~\cite{yao2025deltazip} compresses fine-tuned deltas. JointServe instead jointly allocates distinct-expert weights and session KV.

\emph{Agent execution.}
Parrot and Autellix~\cite{lin2024parrot,luo2025autellix} expose or schedule agent structure; NanoFlow~\cite{zhu2025nanoflow} overlaps pipeline operations. JointServe prefetches the named expert before request arrival; their combination remains untested.

%% file: limitations.tex
\section{Limitations and Future Work}
\label{sec:limitations}

JointServe serves heterogeneous, full-parameter experts on memory-constrained GPU; shared-base experts better suit adapters. Intent extraction assumes that tool-call headers are visible in the caller's output stream, as in structured function calling. The residency objective estimates reuse and recovery cost over a decision window. The measured gains depend on transfer bandwidth, prefill cost, and available overlap, so their magnitude may differ on other hardware or agent traces. Future work can pair cluster replica counts with per-instance residency, use intermediate activations as a third KV recovery source between host reload and recompute, and adapt the currently fixed selection precision and re-solving cadence to observed estimation error.

%% file: conclusion.tex
\section{Conclusion}
\label{sec:conclusion}

Serving multi-model agents on a memory-constrained GPU forces weights
and KV state to compete for the same bytes. JointServe treats that
competition as one scheduling problem. A unified residency formulation
values weight shards and KV blocks on a single recovery-cost scale;
signals read from the agent's own output stream move bytes before
requests exist; and a sharded runtime swaps an expert in 1.83\,s where
a full reload takes 67.76\,s. Across memory budgets and three real
agent workloads, this keeps sessions interactive where whole-model
systems stall, and keeps them feasible where those systems cannot run
at all. As agents assemble ever more specialized models onto fixed
hardware, we believe residency, not raw capacity, becomes the resource
worth scheduling.

%% file: appendix.tex
\section{Joint Selection Guarantee}
\label{app:guarantee}

This appendix establishes the recovery-value guarantee for
the joint residency selection algorithm in the main paper. We first specify the residency model and
its recovery costs, then derive the joint budget decomposition and
bound the losses from candidate compression, value rounding, and
residual-capacity selection. Appendix~\ref{app:oracle} relates this
guarantee to estimated values. Appendix~\ref{app:execution} covers
structured residency choices and transfer readiness.

\subsection{Snapshot model}
\label{app:model}
Fix one decision snapshot. After reserving runtime space and protected
objects, the free capacity is $B_f$. Optional weights $\mathcal W$ and
KV objects $\mathcal K$ have positive sizes $s_i$ and nonnegative
additive values $u_i$. These candidates exclude the protected set.
Discard individually infeasible and zero-value
objects. Write $s(S)=\sum_{i\in S}s_i$, $U(S)=\sum_{i\in S}u_i$, and
$\mathcal X=\{S\subseteq\mathcal W\cup\mathcal K:s(S)\le B_f\}$.
The comparison optimum is $\mathrm{OPT}=\max_{S\in\mathcal X}U(S)$,
using the same candidates, values, protection, and capacity as the
algorithm. The empty set is optimal when no positive-value candidate
remains. The execution layer handles infeasible protection.

Candidate independence means that selecting one object changes
neither another object's size nor its recovery value. It does not
require independent reuse events: expectations of additive recovery
work remain additive even when calls are correlated. Recovery units
must therefore avoid overlapping charges for the same prefix.
Protected value is constant and excluded from the objective.
The analysis holds for any fixed nonnegative input values. With
estimated inputs, replace $u_i$, $U$, and $\mathrm{OPT}$ by
$\hat u_i$, $\hat U$, and $\widehat{\mathrm{OPT}}$, respectively.

\subsection{Recovery costs and value density}
\label{app:cost-density}
The main paper's cost model assigns each object the work
required by its recovery path. Reloading $s$ bytes costs a startup
term plus $s/\mathrm{BW}_{H2D}$. At fixed reuse probability $p$, its
value per byte is $p(\alpha/s+1/\mathrm{BW}_{H2D})$, so startup
overhead contributes more per byte for small transfers.

Recomputation has a different dependence on size. Consider one
full-prefix recovery unit of $n>0$ tokens with $s=\beta_m n$ bytes,
where $\beta_m$ is model $m$'s KV footprint per token. For the fit
$r_m(n)=a_{0,m}+a_mn+b_mn^2$ defined in the main paper, its density is
\begin{equation}
 \begin{split}
 d_m(n)&=\frac{p}{\beta_m}
          \left(\frac{a_{0,m}}{n}+a_m+b_mn\right),\\
 d_m'(n)&=\frac{p}{\beta_m}
          \left(b_m-\frac{a_{0,m}}{n^2}\right).
 \end{split}
 \label{eq:app-density}
\end{equation}
For $p>0$, $a_{0,m}>0$, and $b_m>0$, density decreases up to
$n=\sqrt{a_{0,m}/b_m}$ and increases thereafter. The turning point
depends on the model's fitted costs, and different reuse
probabilities can change the ordering again.

Equation~\eqref{eq:app-density} charges one full-prefix recovery.
When a usable prefix is already present, a missing unit receives
only the incremental recovery work over its token range. A shared
startup cost is charged once, rather than once per physical block.
Taking the cheaper of reload and recompute as an object's cost
assumes that this choice leaves other objects' recovery costs
unchanged. Prefix dependencies instead require the legal options
of Appendix~\ref{app:groups}.

\subsection{Joint budget decomposition}
\label{app:decomposition}

Let $V_W(b)$ and $V_K(b)$ be the exact optima for each object type.
Every joint solution induces a split $b=s(S\cap\mathcal W)$; conversely,
two solutions within complementary budgets have a feasible union.
Consequently,
\begin{equation}
 \mathrm{OPT}=\max_{0\le b\le B_f}\{V_W(b)+V_K(B_f-b)\}.
 \label{eq:app-split}
\end{equation}
The algorithm evaluates discrete representatives of this joint choice.
High-value KV remains KV when accounting for physical occupancy.

\subsection{A feasible value scale}
\label{app:scale}
The value scale $L$ sets the KV threshold and bounds the number of
high-value objects that an optimum can retain.
Let $G$ be the value of the density-ordered feasible prefix, stopping
at the first object that does not fit, and $M$ the best feasible
singleton. Set $L=\max(G,M)$. The fractional knapsack relaxation adds
at most one boundary item to the prefix, so
$L\le\mathrm{OPT}\le G+M\le2L$.
The lower bound holds because both the prefix and the singleton are
feasible. If all candidates fit, the prefix itself is optimal.

\subsection{Compressing high-value KV candidates}
\label{app:construction}
Choose positive budgets $\varepsilon_C,\varepsilon_S,\varepsilon_K$
with sum $\varepsilon<1$, and set $\tau=\varepsilon_KL$.
KV with $u_i>\tau$ enters the high-value set. Partition it into buckets
$[\tau a^j,\tau a^{j+1})$ for $j\ge0$, where
$a=(1-\varepsilon_C)^{-1}$, retaining
at most $h=\lceil2/\varepsilon_K\rceil$ smallest objects per bucket.
Objects with $u_i=\tau$ remain in the low-value set.

An optimal solution contains fewer than $2/\varepsilon_K$ high-value
KV objects, since each contributes more than $\varepsilon_KL$ and
their total value is at most $2L$. If the optimum selects $k$ objects
from a bucket, the $k$ smallest objects are all retained and occupy
no more bytes. Each replacement has at least $1/a=1-\varepsilon_C$
times the value of the object it replaces. Replacing the selected
objects in every bucket thus preserves feasibility and loses at
most an $\varepsilon_C$ fraction of the high-value KV contribution.
This argument uses independently selectable objects. Compression
changes the candidate set searched by the optimizer. Physical
eviction follows only after a residency plan is selected.

\subsection{Minimum-byte representatives}
\label{app:frontier}
Value rounding bounds the dynamic program's state space while
retaining exact byte sizes for capacity checks.
The slow set contains all weights and the retained high-value KV.
Let its size be $m$ and its maximum item value $M_S$. For $m>0$, set
$\lambda=\varepsilon_SM_S/m$ and $q_i=\lfloor u_i/\lambda\rfloor$.
A dynamic program stores the minimum true bytes at each rounded value:
\begin{equation}
 D_j(q)=\min\{D_{j-1}(q),D_{j-1}(q-q_j)+s_j\},
 \label{eq:app-dp}
\end{equation}
with $D_0(0)=0$ and other initial states infinite. A branch with a
negative value index is omitted. Each item is used once, and states
exceeding $B_f$ are removed. Store traceback information
and the representative's unrounded value. For every feasible slow set
$T$, the state $q=\sum_{i\in T}q_i$ has a representative $R$ satisfying
\begin{equation}
 s(R)\le s(T),\qquad U(R)\ge U(T)-m\lambda
                    =U(T)-\varepsilon_SM_S.
 \label{eq:app-representative}
\end{equation}
To see this, $T$ witnesses a state with rounded value
$q=\sum_{i\in T}q_i$, so its minimum-byte representative uses no
more space. Its value is at least $\lambda q$, whereas rounding
loses less than $\lambda$ per selected object.
An empty slow set contributes only the empty representative.

\subsection{Residual selection and joint approximation}
\label{app:proof}
Sort low-value KV by decreasing density. Let $F_L(r)$ be its longest
complete prefix fitting residual capacity $r$. The fractional
relaxation differs by at most one item of value at most $\tau$, giving
$U(F_L(r))\ge V_L(r)-\tau$, where $V_L(r)$ is the exact low-value
optimum. Compare every union $R\cup F_L(B_f-s(R))$ using its unrounded
value and return the best, $S_A$.

\begin{theorem}
\label{thm:joint}
For a fixed independent candidate set with nonnegative additive
values and one capacity constraint, joint selection is feasible and
$U(S_A)\ge(1-\varepsilon)\mathrm{OPT}$.
\end{theorem}
\begin{proof}
Partition an optimum into weights $S_W^*$, high-value KV $S_H^*$,
and low-value KV $S_L^*$. Compression replaces $S_H^*$ without
increasing bytes and loses at most $\varepsilon_CU(S_H^*)$.
Equation~\eqref{eq:app-representative} gives a slow representative
using no more space, with further loss at most $\varepsilon_SM_S$.
The entire $S_L^*$ still fits in its residual capacity, where the
prefix query loses at most $\tau$. This union is one of the compared
candidates, so
\begin{equation}
\begin{split}
 U(S_A)&\ge\mathrm{OPT}-\varepsilon_CU(S_H^*)
                    -\varepsilon_SM_S-\varepsilon_KL\\
       &\ge(1-\varepsilon)\mathrm{OPT}.
\end{split}
\end{equation}
The last inequality uses $U(S_H^*),M_S,L\le\mathrm{OPT}$.
The two parts are disjoint, and each residual query enforces the true
byte budget. Adding back the protected objects preserves feasibility.
\end{proof}

\subsection{Cost and stopping certificate}
\label{app:complexity}
Compression limits the dynamic program's dependence on the number
of KV candidates. For $N$ optional objects and $n_W$ weights, at most
\[J=1+\left\lceil\frac{\ln(2/\varepsilon_K)}{-\ln(1-\varepsilon_C)}\right\rceil\]
high-value buckets give $m\le n_W+hJ$. The rounded state count is
$Q=O(m^2/\varepsilon_S)$. Sorting, the slow program, and binary-search
prefix queries take
$O(N\log N+m^3/\varepsilon_S+Q\log N)$ time.
Rolling values need $O(Q)$ space; a full traceback table needs $O(mQ)$.
For any feasible solution of value $V$, a fractional upper bound
$U_{LP}$ over the original, uncompressed candidates certifies
$V/\mathrm{OPT}\ge V/U_{LP}$ when $U_{LP}>0$. A certificate reaching
the requested ratio permits early termination. Both the certificate
and the approximation theorem concern the input objective.
Appendix~\ref{app:oracle} relates guarantees on estimated values
to actual recovery benefit.

\section{Estimated Values and Oracle Gap}
\label{app:oracle}
The selector uses predicted reuse and recovery costs. This appendix
separates the resulting estimation error from optimization error,
changes between decisions, and the advantage of knowing future
requests. All comparisons use the same feasible family
$\mathcal X$ from Appendix~\ref{app:model}.

\subsection{Error in estimated retention values}
\label{app:estimation}
Let $U$ be the true conditional expected recovery benefit and
$\hat U$ the estimated objective. Set $\alpha=1-\varepsilon$.
If $\sup_{S\in\mathcal X}|U(S)-\hat U(S)|\le E$, then
Theorem~\ref{thm:joint} gives
\begin{equation}
\begin{split}
 U(S_A)&\ge\hat U(S_A)-E
       \ge\alpha\max_S\hat U(S)-E\\
       &\ge\alpha\mathrm{OPT}-(1+\alpha)E.
\end{split}
\label{eq:app-error}
\end{equation}
Thus optimization loses at most an $\varepsilon$ fraction of the true
optimum, while the input error contributes an additive term
$(2-\varepsilon)E$.

Per-object errors $|u_i-\hat u_i|\le e_i$ imply
$E\le\max_{S\in\mathcal X}\sum_{i\in S}e_i\le\sum_i e_i$.
For $u_i=p_ic_i$ and $\hat u_i=\hat p_i\hat c_i$, bounds
$|p_i-\hat p_i|\le\delta_i^p$ and
$|c_i-\hat c_i|\le\delta_i^c$ give
\begin{equation}
 |u_i-\hat u_i|\le c_i\delta_i^p+\hat p_i\delta_i^c.
 \label{eq:app-object-error}
\end{equation}
This connects errors in the Signal Engine and recovery-cost model
to the selection bound. If all object values instead satisfy a
relative error bound $(1-\zeta)u_i\le\hat u_i\le(1+\zeta)u_i$
for $0\le\zeta<1$, the same comparison yields
$U(S_A)\ge\alpha(1-\zeta)\mathrm{OPT}/(1+\zeta)$.

\subsection{Reusing a plan as values change}
\label{app:drift}
A plan computed at time $s$ can be evaluated at a later time $t$
by accounting for accumulated changes in its input values. If its
initial error is at most $E_s$ and
$\sup_{S\in\mathcal X}|U_{j+1}(S)-U_j(S)|\le d_j$, the triangle
inequality gives
\begin{equation}
 E_t\le E_s+\sum_{j=s}^{t-1}d_j.
 \label{eq:app-drift}
\end{equation}
Substituting $E_t$ into Equation~\eqref{eq:app-error} bounds the
value of the reused plan. This argument requires unchanged
candidates, protection, and capacity. A change to those constraints
requires a new feasibility check and selection over the updated
family, as in the main paper's joint budget allocation.

\subsection{Future information and execution loss}
\label{app:information}
Let $R(S,\omega)$ denote the recovery benefit of set $S$ under
future outcome $\omega$, and let $\mathcal I$ be the information
available at decision time. The snapshot objective and a
future-aware oracle have values
\begin{equation}
 \begin{split}
 U(S)&=\mathbb E[R(S,\omega)\mid\mathcal I],\\
 O&=\mathbb E[\max_{S\in\mathcal X}R(S,\omega)\mid\mathcal I].
 \end{split}
 \label{eq:app-oracle-values}
\end{equation}
The oracle chooses after observing the outcome, giving it an
information advantage $I=O-\mathrm{OPT}\ge0$ over the snapshot
optimum.
If execution reduces expected recovery benefit by at most $D$, then
\begin{equation}
 U_{\mathrm{actual}}\ge\alpha(O-I)-(1+\alpha)E-D.
 \label{eq:app-oracle}
\end{equation}
Even exact expected values do not eliminate $I$: when capacity holds
one object and an unannounced request selects uniformly among $N$
objects, advance selection obtains $1/N$ of the clairvoyant benefit.
Thus the static approximation alone specifies no fixed percentage of
a future-aware serving oracle.

\section{Legal States and Execution}
\label{app:execution}
The independent-object model captures the selection decision, while
the runtime must respect legal residency states and transfer
dependencies. This appendix gives an exact KV query for equal-size
objects, a grouped formulation for dependent choices, and the
deadline condition used to reason about streaming readiness. It
also explains how recovery value relates to serving latency.

\subsection{Equal-size KV objects}
\label{app:equal-kv}
When all KV objects are independent and have the same size $a_K$,
the residual-capacity query is exact. Sort them by decreasing value
and let $H(k)$ be the sum of the first $k$ values, with $H(0)=0$.
For $n_K$ KV objects and residual capacity $r$,
\begin{equation}
 V_K(r)=H\!\left(\min\{n_K,\lfloor r/a_K\rfloor\}\right).
 \label{eq:app-top-k}
\end{equation}
Exchanging a selected object for a higher-value unselected object
preserves size and increases value, which proves optimality.
After sorting and prefix-sum construction, each query takes $O(1)$
time.

This variant builds the frontier over weights alone and uses
Equation~\eqref{eq:app-top-k} for the remaining capacity. If its
rounding parameter is $\varepsilon_W$ and the largest feasible
weight value is $M_W$, Equation~\eqref{eq:app-representative} gives
$U(S_A)\ge\mathrm{OPT}-\varepsilon_WM_W
\ge(1-\varepsilon_W)\mathrm{OPT}$. No KV compression or prefix
approximation is needed. With no optional weights, the KV query
alone is optimal.

\subsection{Structured residency choices}
\label{app:groups}
With prefix or mutual-exclusion constraints, enumerate legal
options $(s_{g\ell},u_{g\ell})$ per model or session. Each option
specifies a complete legal choice for that group beyond its
protected state, including a zero option. Groups must have additive
sizes and values, with no shared bytes or duplicated recovery
benefit. Choose exactly one option per group.

For $m_g$ groups and largest positive, individually feasible option
value $M_g$, set $\lambda_g=\varepsilon M_g/m_g$ and
$q_{g\ell}=\lfloor u_{g\ell}/\lambda_g\rfloor$. The minimum-byte
recurrence is
\begin{equation}
 D_g(q)=\min_\ell\{D_{g-1}(q-q_{g\ell})+s_{g\ell}\}.
 \label{eq:app-group-dp}
\end{equation}
Initialize $D_0(0)=0$ and all other states to infinity, omit negative
indices, and discard states exceeding $B_f$. If every option has
zero value, return the zero options. Otherwise, return the
representative with the largest feasible rounded value.

Every feasible solution chooses $m_g$ options, so rounding loses at
most $m_g\lambda_g=\varepsilon M_g\le\varepsilon\mathrm{OPT}$,
where $\mathrm{OPT}$ now ranges over legal grouped choices.
The minimum-byte representative uses no more capacity, yielding
a $(1-\varepsilon)$ approximation. This formulation enforces
dependencies through the options themselves and does not use
cross-group candidate substitution. For $A_g$ explicitly enumerated
options, the dynamic program takes $O(A_gm_g^2/\varepsilon)$ time.
The bound is in the size of that option list; enumerating arbitrary
dependency structures may require exponentially many options.

\subsection{Transfer feasibility and first-use deadlines}
\label{app:deadlines}   
Streaming can start from a resident prefix if the remaining groups
arrive before their first use. Suppose all transfers are available
at intent time, defined as time zero, and use one reserved channel
at a fixed effective rate. Destination space and temporary buffers
must remain feasible throughout. Let $t_j$ be each transfer's
processing time and $d_j$ its first-use deadline, ordered so that
$d_1\le\cdots\le d_J$. A schedule meeting every deadline exists
exactly when
\begin{equation}
 \sum_{j=1}^{k}t_j\le d_k,\qquad k=1,\ldots,J.
 \label{eq:app-deadlines}
\end{equation}
Necessity follows because all transfers with deadlines at most
$d_k$ must finish by that time. For sufficiency, execute transfers
in deadline order; each completion time is its prefix sum and
therefore meets its deadline.

For layer streaming, a deadline is the intent lead time plus the
computation preceding that layer on a stall-free execution path.
KV transfers use the deadline of their first consuming operation.
Checking every prefix matters: total transfer time below total
prefill time can still leave an early layer waiting. The priority
scheduler described in the main paper must also respect write-back
dependencies, later release times, and contention for shared
resources. Execution waits on per-group ready events whenever a
required copy is incomplete. The condition in
Equation~\eqref{eq:app-deadlines} establishes schedulability under
the common-release and reserved-channel assumptions.

\subsection{From recovery value to serving latency}
\label{app:latency}
The selection guarantee bounds avoided recovery work. Its effect on
TTFT and SCT depends on which transfers overlap computation and
which operations remain on the request's critical path. Even in a
serial model with $C=C_0-U$, where $C_0$ is the cost with no optional
recovery benefit, $U(S_A)\ge\alpha\mathrm{OPT}$ implies
\begin{equation}
 C_A\le\alpha C^*+(1-\alpha)C_0,\qquad
 C^*=C_0-\mathrm{OPT},\quad\alpha=1-\varepsilon.
 \label{eq:app-latency}
\end{equation}
The additive term depends on $C_0$, so a value approximation does
not imply the same multiplicative latency bound. Queueing and
overlap further depend on the serving schedule. The main paper's evaluation
therefore measures TTFT and SCT directly.